\documentclass[11pt]{article}

\usepackage{graphicx}
\usepackage[utf8]{inputenc}
\usepackage{amsthm}
\usepackage[fleqn]{amsmath}
\usepackage{amssymb}
\usepackage[lined,boxed]{algorithm2e}
\usepackage{geometry}      
\usepackage{hyperref,fullpage}
\allowdisplaybreaks

\def\R{\mathbb{R}}

\def\gnd{\ensuremath{\mathsf{GND}}\xspace}
\def\opt{\ensuremath{\mathsf{OPT}}\xspace}

\def\P{\mathcal{P}}
\def\sse{\subseteq}

\def\d{\psi}
\def\e{\mathsf{e}}

\newtheorem{theorem}{Theorem}

\newtheorem{cl}{Claim}

\newtheorem{lemma}{Lemma}

\newcommand\ignore[1]{}

\title{Online Generalized Network Design Under (Dis)Economies of Scale}

\author{Viswanath Nagarajan\thanks{Industrial and Operations Engineering Department, University of Michigan. Email: \{viswa,lilyxy\}@umich.edu. Supported in part by NSF grants CCF-1750127 and CMMI-1940766. } \and Lily Wang$^*$}

\begin{document}
\maketitle

\begin{abstract}
We consider a general online network design problem where a sequence of $N$ {\em requests} arrive over time, each of which needs to use some subset of the available {\em resources} $E$. The cost incurred by a resource $e\in E$ is some function $f_e$ of the total load $\ell_e$ on that resource. 
The objective is to minimize the total cost $\sum_{e\in E} f_e(\ell_e)$. We focus on cost functions that exhibit  (dis)economies of scale, that are  of the form $f_e(x) = \sigma_e + \xi_e\cdot x^{\alpha_e}$ if $x>0$ (and zero if $x=0$), 
 where  the exponent $\alpha_e\ge 1$. 
Optimization problems under these functions have received significant recent attention 
 due to applications in energy-efficient computing.   Our main result is a  deterministic online algorithm with tight competitive ratio $\Theta\left(\max_{e\in E} \left(\frac{\sigma_e}{\xi_e}\right)^{1/\alpha_e}\right)$ when $\alpha_e$ is constant for all $e\in E$. This framework is applicable to a variety of network design problems in undirected and directed graphs, including multicommodity routing, Steiner tree/forest connectivity and set-connectivity. In fact, our online competitive ratio even matches the previous-best (offline) approximation ratio for generalized network design. 

\end{abstract}

\setcounter{page}{1}

\section{Introduction}
Network design  problems (involving selecting a subgraph with certain connectivity properties)  are of significant practical and theoretical interest. A classic setting in network design is as follows. There are several requests that need to be routed through a network, where each resource $e$ has a non-decreasing cost-function  $f_e$  that determines the cost $f_e(\ell_e)$ incurred at $e$ as a function of its load $\ell_e$. The objective is to minimize the overall cost $\sum_e f_e(\ell_e)$.

Traditional network design models involve {\em concave} cost-functions. These are cost functions  that exhibit ``economies of scale'', i.e.,  a larger load results in a smaller cost-per-unit-load. This is the setting in  buy-at-bulk network design, that has been studied extensively in approximation and online algorithms~\cite{AA97,ChekuriHKS10,ChakrabartyEKP18}. The most basic problems in this setting are Steiner tree and forest~\cite{AgrawalKR95,GoemansW95}. 

Recent applications in energy-efficient scheduling and routing have motivated the study of cost-functions with ``diseconomies of scale'' \cite{AndrewsAZZ12,MakarychevS18}. Here, larger load results in a larger cost-per-unit-load. These functions capture the energy consumption of network resources that are {\em speed scalable} and adjust their speed in proportion to their load.
The energy consumed at speed/load $x$ grows super-linearly as $x^\alpha$ where the exponent $\alpha>1$. For most technologies, exponent $\alpha$ lies between $1$ and $3$ \cite{AndrewsAZZ12,WiermanAT12}. 

As discussed in \cite{AndrewsAZZ12}, a more accurate model for  energy consumption involves a start-up cost in addition to the  super-linear $x^\alpha$ term. This leads to the cost function:
\begin{equation}\label{eq:energy}
f_e(x) = \left\{ \begin{array}{ll}
0 & \mbox{ if } x=0\\
\sigma_e + \xi_e\cdot x^{\alpha_e} & \mbox{ if } x>0
\end{array}\right.,
\end{equation} 
where the parameters $\sigma_e, \xi_e\ge 0$ and  $\alpha_e\ge 1$ depend on the particular  device (resource). The first term $\sigma_e$ represents the cost incurred in simply keeping the device powered-on but idle and the second term $\xi_e\cdot x^{\alpha_e}$ represents the cost incurred due to speed-scaling. These cost functions exhibit {\em both} economies and diseconomies of scale. Indeed, they appear concave for small values of the load $x$ and convex for large values of the load. So these functions 
are said to exhibit {\em (dis)economies of scale}.  A major challenge in designing algorithms for such cost-functions is that one needs the balance two opposing goals (1) aggregating demands in the concave regime and (2) separating demands in the convex regime. Prior work~\cite{AndrewsAZ16,AntoniadisIKMNP20,KrishnaswamyNPS14} has mainly focused on the special case of {\em uniform} (or related) cost functions where  the $\alpha_e$s and  $\frac{\sigma_e}{\xi_e}$s are uniform across  all resources $e$.

Recently, \cite{EmekKLS20} studied a large class of {\em generalized network design} problems under cost functions of the form~\eqref{eq:energy}, which included routing requests, Steiner tree/forest connectivity and set-connectivity in undirected and directed graphs. The main result in \cite{EmekKLS20} was a  unified approximation framework that provided an  
$O\left(\max_{e} \left(\frac{\sigma_e}{\xi_e}\right)^{1/\alpha_e}\right)$ approximation algorithm assuming only a ``minimum cost oracle'' that can satisfy a {\em single} request at minimum cost.

In this paper, we consider the same class of generalized network design (\gnd) problems as \cite{EmekKLS20}, but in the {\em online setting}. Here, requests arrive over time and each request needs to be (irrevocably) assigned to some resources immediately upon arrival. Our main result is a deterministic online algorithm with competitive ratio $O\left(\max_{e} \left(\frac{\sigma_e}{\xi_e}\right)^{1/\alpha_e}\right)$, which even matches the best approximation ratio known for \gnd.  We also show that no deterministic online algorithm can do better (up to a constant factor).

\subsection{Problem Definition}\label{subsec:prelim}
In the generalized network design (\gnd) problem, we have a set  $E$ of resources and $N$ requests that use these resources. Each request $i\in [N]$ is associated with:
\begin{itemize}
\item a collection $\P_i\sse 2^E$ of ``replies'' where  an algorithm needs to choose some $p_i\in \P_i$ in order to satisfy request $i$. The reply collections may be specified implicitly. 
\item a weight vector $w_i \in \R_{\ge 1}^E$ where request $i$ induces a load of $w_{i,e}$  on each resource $e$ that it uses.  
Note that the  weights on different resources may be unrelated. (The requirement that weights/demands of requests are at least one is common to  all prior work.)
\end{itemize}
Each resource $e\in E$ is associated with an individual cost function $f_e:\R\rightarrow \R$ of the form~\eqref{eq:energy}. We will refer to such functions as {\em (D)oS functions}. 
We emphasize that the parameters $\sigma_e$, $\xi_e$ and $\alpha_e$ may be different across resources. So we can handle networks with heterogenous resources (for example, routers running on different technologies).

A solution is just a choice of reply $p_i\in \P_i$ for each request $i\in [N]$. Then, the load on each resource $e\in E$ is $\ell_e = \sum_{i: e\in p_i} w_{i,e}$. The objective is to minimize the total cost 
 $\sum_{e\in E} f_e(\ell_e)$. 

In the online setting, the requests $i\in [N]$ arrive over time, and the algorithm should choose a reply $p_i\in \P_i$ for each request $i$ immediately upon arrival (which cannot be changed later). As usual, we use competitive analysis to measure the performance of an online algorithm, which is relative  to the offline optimum that knows the entire request sequence upfront.
 
We use $m:=|E|$ to denote the number of resources. For each resource $e\in E$, define $q_e : = (\sigma_e/\xi_e)^{1/\alpha_e}$. Note that $q_e$ is the value of load $x$ at which the two terms $\sigma_e$ and $\xi_e \cdot x^{\alpha_e}$ in the (D)oS cost function $f_e(x)$  become equal. Let $q:=\max_{e\in E} q_e$. Also, let $\alpha:= \max_{e\in E} \alpha_e$ denote the maximum exponent in the (D)oS functions.

\paragraph{Min-cost Oracle} We will assume that the reply-collections $\P_i$ are such that one can find an approximately  min-cost reply efficiently. Formally, we assume that there is a $\tau$-approximation algorithm for the problem $\min_{p\in \P_i} \sum_{e\in p} d_e$ for any request $i\in [N]$ and any scalars $\{d_e\ge 0\}_{e\in E}$. If computational complexity is not a consideration (which is sometimes the case with online algorithms) then this assumption is satisfied trivially with $\tau=1$.

\smallskip \noindent {\bf Example 1 (multicommodity routing).}  
The  resources $E$ are edges in some directed graph $G=(V,E)$. Each request $i\in [N]$ consists of a source $s_i\in V$, destination $t_i\in V$ and demand $d_i\ge 1$. For each $i\in [N]$, the reply-collection $\P_i$ consists of all $s_i-t_i$ paths in $G$, and the weights $w_{i,e}=d_i$ for all $e\in E$. The resulting \gnd instance corresponds to selecting an $s_i-t_i$ routing path carrying $d_i$ units of flow (for each request $i$), so as to minimize the total energy cost of the routing. The min-cost oracle in this case corresponds to the shortest path problem in directed graphs, which admits an exact algorithm: so $\tau=1$.

\smallskip \noindent {\bf Example 2 (set connectivity and set-strong-connectivity).} 
The  resources $E$ are edges in some undirected (resp. directed) graph $G=(V,E)$. Each request $i\in [N]$ consists of a subset $T_i\sse V$ of nodes and demand $d_i\ge 1$. The  reply-collection $\P_i$ consists of all edge-subsets that induce  a connected (resp. strongly connected) subgraph containing $T_i$. The  weights $w_{i,e}=d_i$ for all $e\in E$.  The resulting \gnd instance corresponds to selecting an overlay network for each terminal-set $T_i$ that can support $d_i$ units of flow. The min-cost oracle for the undirected case corresponds to the Steiner  tree problem: so we have $\tau=1.39$~\cite{ByrkaGRS13}. In the directed case, the oracle is the strongly connected Steiner subgraph problem, for which  we have (i) $\tau=k^\epsilon$ for any constant $\epsilon>0$ in polynomial time~\cite{CharikarCCDGGL99} or (ii) $\tau=O(\frac{\log^2k}{\log\log k})$ in {\em quasi-polynomial} time~\cite{GLL19,GhugeN20}. Here $k=\max_i |T_i|$ is the maximum number of terminals in any request. 

\subsection{Our Results and Techniques}\label{subsec:results}
Our main result is the following:
\begin{theorem}\label{thm:main-1}
There is a polynomial time $O(q\tau + (\e \alpha\tau)^\alpha)$-competitive deterministic online algorithm for \gnd assuming a $\tau$-approximation algorithm for the min-cost oracle. 
\end{theorem}
Above, $\e\approx 2.718$ is the base of the natural logarithm. The running time of this algorithm is $O(Nm + N\cdot \Phi(m))$ where $\Phi(m)$ is the time taken by the min-cost oracle. Note that when $\tau=1$, we obtain a competitive ratio of $O(q+(\e\alpha)^\alpha)$. 

To the best of our knowledge, previous online algorithms for \gnd were restricted to the case of multicommodity routing in undirected graphs with uniform edge-cost functions~\cite{AntoniadisIKMNP20}. Our result provides a unified framework to address various types of requests (including Steiner and set-connectivity) in both undirected and directed graphs. Moreover, this is the first competitive ratio (even in the previously-studied setting \cite{AntoniadisIKMNP20}) that does not grow with the network size or the number of requests. Finally, our 
result also applies to {\em non-uniform} cost functions: in this  setting, no online algorithm was known even for single-commodity routing with edge costs.  

 As noted earlier, our competitive ratio matches the $O_\alpha(q\tau + \tau^\alpha)$ approximation algorithm for \gnd obtained in \cite{EmekKLS20}.\footnote{The $O_\alpha$ notation treats $\alpha$ as constant and suppresses   factors that depend on $\alpha$.} Even when used in the offline setting, our algorithm has several advantages. First, the dependence on $\alpha$ in the approximation ratio is better: we obtain a factor of $(\e \alpha)^\alpha=\e^{\alpha (1+\ln \alpha)}$ whereas the previous algorithm had a $3^{\alpha^2}$ factor \cite{EmekKLS20-comm}. Second, our algorithm is deterministic whereas the previous algorithm was randomized. Third, our running time is  better. Fourth, our algorithm itself is very simple and (arguably) simpler to analyze. 
 
To prove  Theorem~\ref{thm:main-1}, we first show that any (D)oS function $f_e(x)$ of the form \eqref{eq:energy} can be well-approximated by a weighted sum of power functions of form $h_e(x)=\eta_e\cdot x + \xi_e\cdot x^{\alpha_e}$. This reduction loses a factor of $2(\sigma_e/\xi_e)^{1/\alpha_e}$ in the objective. This allows us to  then focus on the \gnd problem under (non-uniform) power cost functions, which is a convex objective. 

For \gnd under power cost functions,
 if we were only interested in an offline approximation algorithm, we could use the approach in \cite{MakarychevS18} that was based on a convex relaxation and rounding to obtain an $A_\alpha$-approximation algorithm for \gnd (assuming $\tau=1$). Here, $A_\alpha \approx (\frac{\alpha}{\ln(1+ \alpha)})^\alpha$ is the fractional Bell number. This approach however does not work in the online setting. Instead, we use a more direct approach  motivated by work on online load balancing with $\ell_p$-norms~\cite{AwerbuchAGKKV95}. For each request $i$, our algorithm basically selects the reply in $\P_i$ that results in the smallest increase in the objective. (The actual algorithm involves tracking a modified objective function.) We analyze our algorithm using the online primal-dual method for convex programs. The idea is to (1) write a convex relaxation for \gnd and its dual, and (2) upper bound the (integral) primal objective by some factor $\rho$ times the dual objective. By weak duality, we then obtain a competitive ratio of $\rho$. 
 
 There have been a number of recent papers using the online primal-dual approach for convex programs (see \S\ref{subsec:related} for more details). The work closest to ours is 
  \cite{GuptaKP12}, where an $O(\alpha)^\alpha$-competitive algorithm was obtained for the special case of \gnd with uniform $\alpha$ power cost functions and multicommodity routing requests. Our approach is more general as it can handle a much wider class of requests and non-uniform $\alpha_e$ powers. From a technical perspective, while our primal convex program is the natural extension of that in \cite{GuptaKP12} (for multicommodity routing),  we use a different (re)formulation of the dual program and also set dual variables differently. Our dual formulation is easier to reason about, and hence allows for a clean analysis even in  more general settings. 
 
Implementing the above approach directly leads to an $O(q(\e \alpha \tau)^\alpha)$-competitive algorithm for \gnd using a $\tau$-approximate min-cost oracle. To obtain the more refined guarantee in Theorem~\ref{thm:main-1}, we improve both steps above. In the reduction from (D)oS functions $f_e(x)$ to power functions $h_e(x)$,  we show that the factor $q_e$ loss only affects the linear term in $h_e(x)$.  Then, in the online algorithm for \gnd under power functions, we show that the greedy objective can be further modified to ensure a stronger $O(\tau)$ competitive ratio for the linear terms, while the non-linear terms incur an $O((\e \alpha \tau)^\alpha)$ competitive ratio.

We also provide a nearly matching lower bound for online \gnd:
\begin{theorem}\label{thm:main-2}
Every deterministic online algorithm for \gnd has competitive ratio $\Omega\left( q +(1.44 \alpha)^\alpha\right)$.
\end{theorem}
As usual with online lower bounds, this is information-theoretic and independent of computational requirements. So this nearly matches the $O(q+(\e \alpha)^\alpha)$ competitive ratio from Theorem~\ref{thm:main-1} when $\tau=1$. The lower bound instance involves single-commodity routing requests in directed graphs. The $\Omega(q)$ part of the lower bound relies on a construction similar to the online directed Steiner tree lower bound~\cite{FaloutsosPS02}. The $\Omega((1.44 \alpha)^\alpha)$ part of the lower bound follows from the corresponding result for online load balancing  with $\alpha^{th}$ power of loads~\cite{Caragiannis08}. 

Finally, we can also extend our main result to a larger class of functions called {\em real exponent polynomials} (REP) that were studied in \cite{EmekKLS20}. These have the form
\begin{equation}\label{eq:REP}
\bar{f}_e(x) = \left\{ \begin{array}{ll}
0 & \mbox{ if } x=0\\
\sigma_e + \sum_{j=1}^q \xi_{e,j} \cdot x^{\alpha_{e,j}} & \mbox{ if } x>0
\end{array}\right.,
\end{equation} 
where the parameters $\sigma_e, \xi_{e,1},\cdots \xi_{e,q}\ge 0$ and the exponents $\alpha_{e,j}\ge 1$. 
\begin{theorem}\label{thm:main-3}
There is a polynomial time $O(Q\tau + (\e \alpha\tau)^\alpha)$-competitive deterministic online algorithm for \gnd under REP cost functions assuming a $\tau$-approximation algorithm for the min-cost oracle. Here $Q=\max_{e\in E} \min_{j\in [q]}  (\sigma_e/\xi_{e,j})^{1/\alpha_{e,j}}$. 
\end{theorem}
The idea here is to reduce any \gnd problem with REP costs into another instance with (D)oS cost functions of form~\eqref{eq:energy} but with more resources.  
\subsection{Related Work}\label{subsec:related}

Most of the prior work in network design under (D)oS cost functions has focused on multicommodity routing requests with uniform weights (i.e., $w_{i,e}=d_i$ for all resources $e$ and requests $i$). \cite{AndrewsAZZ12} were the first to study this model and obtained an $O(q\cdot \log^{\alpha-1} D)$-approximation algorithm where $D=\max_{i=1}^N d_i$ is the maximum weight. When $\sigma_e=0$ for all resources $e$ (in which case the objective is a weighted sum of power functions), \cite{MakarychevS18} obtained an improved $A_\alpha$-approximation algorithm. These results apply to  undirected as well as directed graphs.

Further results are known for multicommodity routing in {\em undirected graphs} in the special case of {\em uniform} cost functions, where $f_e(x) =c_e\cdot f(x)$ for a common (D)oS function $f(x)$.  When costs are incurred on edges, \cite{AndrewsAZ16} obtained a poly-logarithmic $O(\log^{O(\alpha)} N)$-approximation algorithm, and \cite{AntoniadisIKMNP20}  later improved the approximation ratio to  $O(\log^{ \alpha } N)$. When costs are incurred on nodes (which is harder than the edge-version), \cite{KrishnaswamyNPS14} obtained an $O(\log^{O(\alpha)} N)$-approximation algorithm. All these results rely crucially on the uniformity of the cost function. In particular, they use the fact that it is best to aggregate $q=(\sigma/\xi)^{1/\alpha}$ units of demand, after which the aggregated demands can be routed in a ``well separated'' manner. It is unclear if these techniques can be used for non-uniform costs as the ``aggregate demand'' quantity for different resources is different (it is $q_e = (\sigma_e/\xi_e)^{1/\alpha_e}$ for each resource $e$). Furthermore, these results    relied on cut-sparsification and small flow-cut gaps, which do not extend to directed graphs. In fact, the directed Steiner forest problem (which is a special case of \gnd) is hard to approximate better than $\Omega(2^{\log^{1-\epsilon} N})$ for any constant $\epsilon>0$ \cite{DodisK99}. We note that the parameter $q\approx N$ for \gnd instances corresponding to Steiner forest: so we cannot expect an approximation ratio much better than $poly(q)$ for \gnd.  In fact, any $o(\sqrt{q})$-approximation algorithm for \gnd would improve on the best approximation ratio known for directed Steiner forest~\cite{ChekuriEGS11,FeldmanKN12}.

As mentioned earlier, \cite{EmekKLS20} considered the much wider class of \gnd problems, and obtained an $O(q)$-approximation algorithm. As discussed in \cite{EmekKLS20}, their result extends prior work involving (D)oS cost functions in several ways: unrelated weights, non-uniform cost functions, strongly polynomial runtime etc. Our result inherits all these advantages even in the online setting. The technique in \cite{EmekKLS20} was based on the ``smoothness'' toolbox from \cite{Roughgarden15}. Our approach (discussed above) is completely different, and leads to a much simpler algorithm.

In the online setting, \cite{AntoniadisIKMNP20} obtained an $\tilde{O}(\log^{3\alpha +1}N)$-competitive randomized algorithm for multicommodity routing in undirected graphs with uniform cost functions on edges and uniform weights.  This ratio is incomparable to the $O(q+(\e \alpha)^\alpha)$ deterministic online ratio that we obtain (even in more general settings).  
When $\sigma_e=0$ for all resources $e$ and all $\alpha_e$ are uniform, $O(\alpha)^\alpha$-competitive online algorithms were known for load balancing~\cite{AwerbuchAGKKV95} and multicommodity routing \cite{GuptaKP12}. Our algorithm can be seen as a natural extension of these results to the  setting of \gnd. \cite{AwerbuchAGKKV95} used a  potential-function analysis that appears hard to extend  to non-uniform $\alpha_e$s.  As discussed in \S\ref{subsec:results}, though our approach as well as \cite{GuptaKP12} are  based on the online primal-dual method, there are important differences as well.

The online primal-dual method (see the survey~\cite{BuchbinderN09}) is a very general technique that has led to several strong results in online algorithms. Typically, this approach is applied with covering/packing linear-program relaxations, e.g. \cite{AlonAABN09,AlonAABN06,BansalBN12}. However, a number of recent papers, e.g. \cite{GuptaKP12,AnandGK12,DevanurH18,AzarBCCCGHKNNP16,NagarajanS17,HuangK19}, have extended this to the setting of covering programs with convex objectives. Our result adds to this line of work. Although our fractional relaxation is a ``convex covering program'' as studied in \cite{AzarBCCCGHKNNP16}, we cannot use the general-purpose algorithm presented there because the number of variables in our relaxation for \gnd is exponential: the competitive ratio in \cite{AzarBCCCGHKNNP16} is logarithmic in the number of variables. We note however that our idea of setting  dual variables based on the gradient of the primal objective (at the final solution) was  partly motivated from \cite{AzarBCCCGHKNNP16}. 

\subsection{Paper Outline}
We start with the reduction from (D)oS cost functions to weighted power functions in \S\ref{sec:redn}. In \S\ref{sec:frac} we provide a fractional online algorithm for the natural convex relaxation of \gnd under power cost functions. Then, in \S\ref{sec:int} we extend this to an integral online algorithm. \S\ref{sec:overall} puts things together and finishes the proofs of Theorems~\ref{thm:main-1} and \ref{thm:main-3}. Finally, \S\ref{sec:LB} provides the online lower bounds (Theorem~\ref{thm:main-2}). 

\section{Reducing (D)oS Functions to Weighted Power Functions}\label{sec:redn}
We first make the simple but useful observation that any cost-function $f_e$ of the form~\eqref{eq:energy} can be approximated by a {\em convex power function}, at the loss of a multiplicative factor $2q_e$, where $q_e : = (\sigma_e/\xi_e)^{1/\alpha_e}$. To this end, define for each $e\in E$, a new function
\begin{equation}\label{eq:convex-fn}
h_e(x) \,:= \,\xi_e q_e^{\alpha_e-1}\cdot x + \xi_e\cdot x^{\alpha_e},\quad \mbox{for all }x\ge 0.
\end{equation}
\begin{lemma}\label{lem:fn-apx} For each $e\in E$ and $x\in \{0\}\cup \R_{\ge 1}$, we have 
$$\frac12 \cdot h_e(x)\le f_e(x)\le \max\{q_e,1\} \cdot \xi_e q_e^{\alpha_e-1}\cdot x + \xi_e\cdot x^{\alpha_e} \le \max\{q_e,1\} \cdot h_e(x).$$
\end{lemma}
\begin{proof}
At $x = 0$ the inequalities trivially hold. So we assume $x\ge 1$ in the rest of the proof. For the first inequality, we divide it into two cases. If $x < q_e$, then 
\[
h_e(x) = \xi_e q_e^{\alpha_e-1}\cdot x + \xi_e\cdot x^{\alpha_e} \leq \xi_e q_e^{\alpha_e} + \xi_e\cdot x^{\alpha_e} =  \sigma_e +  \xi_e\cdot x^{\alpha_e}= f_e(x)
\]
If $x \geq q_e$, then 
\[
h_e(x) = \xi_e q_e^{\alpha_e-1}\cdot x + \xi_e\cdot x^{\alpha_e} \leq 2 \xi_e x^{\alpha_e} \leq 2 (\xi_e x^{\alpha_e} + \sigma_e) = 2f_e(x)
\] 
For the second inequality, we have $$\max\{1,q_e\}\cdot \xi_e q_e^{\alpha_e-1}\cdot x + \xi_e\cdot x^{\alpha_e}  \ge \xi_e q_e^{\alpha_e}\cdot x +  \xi_e\cdot x^{\alpha_e} = \sigma_e x +   \xi_e\cdot x^{\alpha_e} \geq \sigma_e + \xi_e\cdot x^{\alpha_e} = f_e(x),$$
where the second inequality uses $x\ge 1$.  
\end{proof}
Recall  that $q:=\max_{e\in E} q_e$. By Lemma~\ref{lem:fn-apx}, at the loss of factor $2\max\{q,1\}$, it suffices to solve the \gnd problem under power cost functions, where each resource $e\in E$ has a cost function of the form $g_e(x)=c_e\cdot x^{\alpha_e}$ (see details  in \S\ref{sec:overall}). In the next two sections, we provide  online algorithms for \gnd under {\em weighted power functions}. 

\section{Fractional Online Algorithm}\label{sec:frac}
We consider the following convex program relaxation for \gnd, denoted $(P)$. 
\begin{align}
  \min\quad&\sum_{e\in E} c_e\cdot \left( \sum_{i=1}^N w_{i,e} \sum_{p\in \P_i: e\in p} x_{i,p}\right)^{\alpha_e}\notag\\
  \mbox{s.t.}\quad& \sum_{p\in \P_i}x_{i,p} \ge 1,\qquad \forall i\in [N]\label{cons:primal}\\
  &\mathbf{x} \ge  \mathbf{0}.\notag
\end{align}
Note that all constraints are of ``covering type'' and the objective is convex. However, there are an exponential number of variables as the replies $\P_i$ are implicitly specified. We will solve this program approximately using the online primal-dual method. First, we provide a continuous time online algorithm, that is easier to describe and analyze (Theorem~\ref{thm:frac-online}). Then, we explain how to obtain a polynomial time implementation at a small loss in the competitive ratio (\S\ref{app:poly-time}).

Let $E_1=\{e\in E: \alpha_e=1\}$. 
The dual of  convex program $(P)$ is below, denoted $(D)$.
  \begin{align}
  \max\quad&\sum_{i=1}^N y_i \,-\, \sum_{e\in E\setminus E_1} \frac{ c_e\alpha_e}{\beta_e} \cdot z_e^{\beta_e} \notag\\
  \mbox{s.t.}\quad& \sum_{e\in p}w_{i,e} c_e \alpha_e\cdot z_e \ge y_{i},\qquad \forall p\in \P_i,\, \forall i\in [N] \label{cons:dual}\\
  & z_e \le 1, \qquad\qquad\qquad\qquad\forall e\in E_1 \label{cons:d-UB}\\
  &\mathbf{y}, \mathbf{z} \ge  \mathbf{0}.\notag
\end{align}
  Above, for each $e\in E\setminus E_1$, value $\beta_e> 1$ is the conjugate of $\alpha_e$, i.e. $\frac{1}{\alpha_e}+\frac{1}{\beta_e}=1$. Note that there are no terms in the dual objective  corresponding to $e\in E_1$.  We derive this dual in Appendix~\ref{app:derive-dual}. It turns out that strong duality holds for this primal-dual pair. However, we will only use weak duality, which is proved below.
\begin{lemma}\label{lem:weak-duality}
For any primal $x\in (P)$ and dual $(y,z)\in (D)$ solutions,
$$ \sum_{e\in E} c_e\cdot \left( \sum_{i=1}^N w_{i,e} \sum_{p\in \P_i: e\in p} x_{i,p}\right)^{\alpha_e} \,\,\ge \,\, \sum_{i=1}^N y_i \,-\, \sum_{e\in E\setminus E_1} \frac{ c_e\alpha_e}{\beta_e} \cdot z_e^{\beta_e}.$$
\end{lemma} 
\begin{proof}
For easier notation, let $\ell_e:=\sum_i \sum_{p\in \P_i: e\in p} w_{i,e}\cdot x_{i,p}$ be the fractional load on each $e\in E$. For each   $e\in E_1$, let $\beta_e=\infty$: note that $\frac1{\beta_e} z_e^{\beta_e}=0$ as $z_e\le 1$.  We will show that 
$$\sum_i y_i \le \sum_{e\in E} \alpha_e c_e\cdot \left( \frac{1}{\alpha_e}\cdot \ell_e^{\alpha_e} + \frac{1}{\beta_e}\cdot z_e^{\beta_e} \right),$$ which would prove the lemma.  Indeed, we have:
\begin{align}
\sum_i y_i & \le \sum_i \left(\sum_{p\in \P_i} x_{i,p} \right) \cdot y_i   \le  \sum_i  \sum_{p\in \P_i} x_{i,p} \cdot \left(\sum_{e\in p} w_{i,e}c_e\alpha_e\cdot z_e\right) \label{eq:weak-dual-1}\\
& = \sum_e c_e \alpha_e  \cdot z_e \left(\sum_i \sum_{p\in \P_i: e\in p} w_{i,e} \cdot x_{i,p}\right) = \sum_e c_e \alpha_e  \cdot z_e \cdot \ell_e \label{eq:weak-dual-2}\\
&\le \sum_{e\in E_1} c_e \cdot \ell_e +  \sum_{e\in E\setminus E_1} c_e \alpha_e  \cdot z_e \cdot \ell_e   \le \sum_e c_e \alpha_e  \left( \frac{1}{\alpha_e}\cdot \ell_e^{\alpha_e} +  \frac{1}{\beta_e}\cdot z_e^{\beta_e} \right).\label{eq:weak-dual-3}
\end{align}
Above, the first inequality in \eqref{eq:weak-dual-1} is by constraint~\eqref{cons:primal} and non-negativity, and the last inequality in \eqref{eq:weak-dual-1} is by constraint~\eqref{cons:dual}. The equality in \eqref{eq:weak-dual-2} is by interchanging summation. The first inequality in \eqref{eq:weak-dual-3} is by constraint~\eqref{cons:d-UB} and the last inequality is by Young's inequality, which says $A\cdot B \le \frac{1}{\alpha}\cdot A^\alpha +  \frac{1}{\beta}\cdot B^\beta$ for any $A,B\ge 0$ and $\alpha,\beta > 1$ with $\frac{1}{\alpha}+\frac{1}{\beta}=1$. This completes the proof.
\end{proof}

\begin{algorithm}
\DontPrintSemicolon 
 Upon arrival of request $i$, do the following.\;
\For{each continuous time $t\in [0,1]$}{
    Choose reply $p^*\in \P_i$ using the min-cost oracle under  costs $d_e= \alpha_e c_e \cdot \ell_e^{\alpha_e-1}\cdot w_{i,e}$ for each $e\in E$, where $\ell_e=\sum_i \sum_{p\in \P_i: e\in p} w_{i,e}\cdot x_{i,p}$ is the current fractional load on $e$.\; 
Raise primal variable $ x_{i,p^*}$ at rate one, i.e. $\frac{\partial}{\partial t} x_{i,p^*}=1$.\;
}
\caption{Fractional online algorithm for $(P)$}
\end{algorithm}

\begin{theorem}\label{thm:frac-online}
The fractional online algorithm has competitive ratio at most $\alpha^\alpha$ where $\alpha=\max_{e\in E} \alpha_e$.
\end{theorem}
  \begin{proof}
  The proof is by dual fitting: we will provide a feasible dual solution $(y,z)$ and show that the online primal solution $\bar{x}$ has objective at most $\alpha^\alpha$ times the dual objective. Combined with Lemma~\ref{lem:weak-duality}, this would imply the theorem. 
  
Let $\bar{\ell}_e= \sum_i \sum_{p\in \P_i: e\in p} w_{i,e}\cdot \bar{x}_{i,p}$ be the final load on each $e\in E$. Let $\delta\in (0,1]$ be some parameter, and define the dual solution:
  $$z_e = \delta\cdot \bar{\ell}_e^{\alpha_e-1},\qquad \forall e\in E.$$
  $$y_i = \min_{p\in \P_i} \sum_{e\in p} w_{i,e}c_e\alpha_e\cdot z_e,\qquad \forall i\in [N].$$
Note that dual-constraint~\eqref{cons:d-UB} is satisfied as $z_e=\delta\le 1$ for all $e\in E_1$. Moreover, \eqref{cons:dual} is satisfied by definition of $y$. So $(y,z)$ is a feasible dual solution. For each request $i$, let $q_i\in \P_i$ denote the reply that achieves the minimum cost in the definition of $y_i$ above.

We now relate the primal objective $\bar{P} = \sum_e c_e\cdot \bar{\ell}_e^{\alpha_e}$ with the dual objective $D$, by showing:
\begin{equation}
\label{eq:primal-dual-frac}
D  \,\,\ge \,\, \left(\delta -(\alpha-1)\cdot \delta^{\frac{\alpha}{\alpha  -1}}\right)\cdot \bar{P} 
\end{equation}
Consider the algorithm when some request $i$ arrives. For each time $t\in [0,1]$, if $p^*\in\P_i$ is the current reply and $\{\ell_e\}_{e\in E}$ denotes the current loads, then by the primal update:
\begin{align*}
\frac{\partial}{\partial t} \bar{P} & = \sum_{e\in p^*} c_e\alpha_e   \left( \sum_{i=1}^N w_{i,e} \sum_{p\in \P_i: e\in p} x_{i,p}\right)^{\alpha_e-1} w_{i,e}  = \sum_{e\in p^*} w_{i,e} c_e\alpha_e\cdot  \ell_e^{\alpha_e-1} \le \sum_{e\in q_i} w_{i,e} c_e\alpha_e\cdot  \ell_e^{\alpha_e-1}  \\
& \le  \sum_{e\in q_i} w_{i,e} c_e\alpha_e\cdot \bar{\ell}_e^{\alpha_e-1} =  \frac{1}{\delta}\cdot y_i.
\end{align*}
Above, the first inequality is by the choice of the current reply $p^*$ at time $t$, the second inequality is by monotonicity of the primal solution $x$ over time, and the last equality is by the choice of the dual value $y_i$. It follows that the increase in $\bar{P}$ due to request $i$ is at most $\frac{y_i}{\delta}$. Adding over all $i$,  
$$\bar{P} \le \frac{1}{\delta}\sum_{i=1}^N y_i.$$
Now, consider the contribution of the $z$-variables to the dual objective:
$$\sum_{e\in E\setminus E_1} \frac{c_e\alpha_e}{\beta_e}\cdot z_e^{\beta_e} = \sum_{e\in E\setminus E_1} \delta^{\beta_e}\frac{c_e\alpha_e}{\beta_e} \left(\bar{\ell}_e^{\alpha_e-1}\right)^{\beta_e}= \sum_{e\in E\setminus E_1} \delta^{\beta_e} c_e(\alpha_e-1)  \bar{\ell}_e^{\alpha_e}\le \delta^{\frac{\alpha}{\alpha-1}}(\alpha-1) \sum_{e\in E\setminus E_1}c_e \bar{\ell}_e^{\alpha_e}.  $$
The equalities use the fact that $\frac{1}{\beta_e}=1-\frac{1}{\alpha_e}$. The inequality above uses that $\delta\le 1$ and $\beta_e = 1+\frac{1}{\alpha_e-1}\ge 1+\frac{1}{\alpha-1}$ for all $e$. Finally,   the right-hand-side above is at most  $\delta^{\frac{\alpha}{\alpha-1}}(\alpha-1) \cdot \bar{P}$. Therefore, the dual objective is:
$$D = \sum_{i=1}^N y_i - \sum_{e\in E\setminus E_1} \frac{c_e\alpha_e}{\beta_e}\cdot z_e^{\beta_e} \ge \delta\cdot \bar{P} - \delta^{\frac{\alpha}{\alpha-1}}(\alpha-1) \cdot \bar{P},$$
which proves \eqref{eq:primal-dual-frac}. Finally, choosing $\delta=1/\alpha^{\alpha-1}$, we obtain $\bar{P}\le \alpha^\alpha\cdot D$. 
\end{proof}

\subsection{Polynomial Time Algorithm}\label{app:poly-time}
  To make the previous (continuous-time) algorithm run in polynomial time, we show how to reduce the number of queries to the min-cost oracle. The main idea is to  perform  a new query whenever the  cost (under  current loads) of the current reply  increases by a factor $1+\epsilon$, where $\epsilon>0$ is a constant. Recall that the cost-function under loads $\{\ell_e\}_{e\in E}$ is
  $d_e= \alpha_e c_e \cdot \ell_e^{\alpha_e-1}\cdot w_{i,e}$ for all $e\in E$. We also artificially increase the initial load on every resource to be $\eta\rightarrow 0$ rather than zero. 
Below, we will ensure that (1) the number of queries to the min-cost oracle is polynomial and (2) the competitive ratio is still roughly $\alpha^\alpha$.
 
 Recall that $m=|E|$ is the number of resources, and $\alpha = \max_e \alpha_e$. By scaling costs, we can assume (without loss of generality) that $c_e\ge 1$ for all $e\in E$. Moreover, recall that all weights $w_{i,e}\ge 1$. Let $B$ denote the maximum cost/weight in the instance. 
  Let  $p^*$ denote the current reply at any point of the new algorithm. 
 Note that  reply $p^*$ is always a $1+\epsilon$ approximately min-cost reply under the current cost function $\{d_e\}$.

 For the competitive ratio,  consider how the analysis changes when the  reply $p^*$ is only guaranteed to be a  $1+\epsilon$ approximate reply (rather than min-cost). The increase in the primal objective due to request $i$ is then at most $\frac{(1+\epsilon) y_i}{\delta}$. Adding over all $i$, we obtain $\bar{P} - I \leq \frac{1+\epsilon}{\delta} \sum y_i$, where $I=\sum_{e\in E} c_e \eta^{\alpha_e}$ is the initial primal objective and $\bar{P}$ is the final objective. Note that $I\le mB\eta$. As before, the dual objective is bounded as \[D \geq \frac{\delta}{1+\epsilon} \cdot (\bar{P} -I) - \delta^\frac{\alpha}{\alpha - 1}(\alpha - 1) \cdot \bar{P} \ge \left(\frac{\delta}{1+\epsilon}   - \delta^\frac{\alpha}{\alpha - 1}(\alpha - 1) \right)\cdot \bar{P} -I\]
Choosing $\delta$ as $(\frac{1}{\alpha (1+\epsilon)})^{\alpha-1}$ to maximize the coefficient on $\bar{P}$, we obtain $\bar{P} \leq ((1+\epsilon)\alpha)^{\alpha} \cdot (D + I)$. By weak duality, we know that $D\le \opt$ the optimal fractional value. We now bound $I\le mB\eta$ in terms of \opt. In any fractional solution $\{x_{i,p}\}$, for any request $i$, we  have  
$$\sum_e \sum_{p \in P_i: e\in p} x_{i,p} = \sum_{p\in \P_i} |p|\cdot x_{i,p}\ge  \sum_{p\in \P_i}  x_{i,p}\ge 1.$$
Averaging over all resources, some $e\in E$ has $ \sum_{p \in P_i: e\in p} x_{i,p}\ge \frac1m$, which means its load is  at least $\frac{1}{m}$ (as all weights are at least one). So  cost of any fractional solution is at least $\frac{1}{m^\alpha}$. It now follows that $I\le mB\eta \le m^{\alpha+1}B\eta\cdot \opt$. Choosing $\eta = \frac{\epsilon}{m^{1+\alpha} B}$,  the primal objective $\bar{P}\le (1+\epsilon)^{\alpha+1}\alpha^\alpha\cdot \opt$.

 We now bound  the number of queries. Note that the min-cost of any reply is at least $\eta^{\alpha-1}$ as all the loads are initially $\eta$. Moreover, any load  $\ell_e\le NB$ which implies that the maximum cost of any reply is at most $\alpha m N^{\alpha-1}B^{\alpha + 1}$. As we make a new query only when the current cost of $p^*$  increases by a factor $1+\epsilon$, the number of queries is at most 
 $$\log_{1+\epsilon} \left(\frac{\alpha m N^{\alpha-1}B^{\alpha+1}}{ \eta^{\alpha-1}}\right) = O\left(\alpha^2 \log(mnB)\right),$$
 where we used the above choice of $\eta$. So the number of queries is polynomial.

  \section{Integer Online Algorithm}\label{sec:int}
 
  We now provide an integral online algorithm for \gnd. It is well-known (see e.g. \cite{AndrewsAZZ12}) that the convex relaxation $(P)$ used in \S\ref{sec:frac} has a polynomially large integrality gap even for single-commodity routing on undirected graphs. To get around this, we use an idea from \cite{AzarE05} for load balancing, by adding additional {\em linear terms} corresponding to the $\alpha_e^{th}$ power of loads from individual requests. 
  Let $\rho\ge 1$ be a parameter to be set later. Upon the arrival of request $i$, we do the following:
  \begin{itemize}
\item Choose reply $p_i\in \P_i$ using the min-cost oracle under the costs 
\begin{equation}\label{eq:cost-integer-online}
\d_e= \alpha_e c_e \cdot \ell_e^{\alpha_e-1}\cdot w_{i,e} \,+\, \frac{\rho}{\e^{\alpha}} \cdot c_e \alpha_e w_{i,e}^{\alpha_e}, \mbox{ for each }e\in E,
\end{equation} 
where $\ell_e \,:=\, \sum_{j<i : e\in p_j} w_{j,e}$ is the current load on $e$. 
\end{itemize}
\begin{theorem}\label{thm:online}
The online \gnd algorithm has competitive ratio at most $2(\e\alpha)^{\alpha}$ where $\alpha=\max_{e\in E} \alpha_e$. 
\end{theorem}
We prove this result in the rest of this section. 
Let $A_e$ denote the final load on each resource $e\in E$. The online algorithm's objective is then $A:=\sum_{e} c_e\cdot A_e^{\alpha_e}$.

We will use a different (stronger) convex relaxation for \gnd and relate $A$ to the new relaxation. The new relaxation has the same constraints in $(P)$ but the objective is now:
\begin{equation}
    \label{eq:new-obj}
\sum_{e\in E} c_e\cdot \left( \sum_{i=1}^N w_{i,e} \sum_{p\in \P_i: e\in p} x_{i,p}\right)^{\alpha_e}\,\,+\,\,\sum_{e\in E} \frac{c_e \alpha_e }{\e^{\alpha}} \cdot \sum_{i=1}^N w_{i,e}^{\alpha_e}\sum_{p\in \P_i: e\in p} x_{i,p} \end{equation}
\begin{lemma} \label{lem:new-opt}
The optimal value of the new  convex program with objective~\eqref{eq:new-obj} is at most $(1+\alpha\e^{-\alpha})\cdot \opt$, where \opt is the optimal value of the (integral) \gnd instance. 
\end{lemma}
\begin{proof}
Consider an  optimal solution to \gnd with objective \opt. We set a corresponding solution for ($P$) by setting $x_{i,p}$ to 1 if $p$ is the reply used to satisfy request $i$ and 0 otherwise. Using the fact that each $x_{i,p} $ is either 0 or 1,  we have for each $e$,  \[\sum_{i=1}^N w_{i,e}^{\alpha_e}\sum_{p\in \P_i: e\in p} x_{i,p} \leq \left(\sum_{i=1}^N w_{i,e} \sum_{p\in \P_i: e\in p} x_{i,p}\right)^{\alpha_e}\]  So, the objective of the new relaxation is at most \[(1+\frac{\alpha}{\e^\alpha})\sum_{e\in E} c_e\cdot \left( \sum_{i=1}^N w_{i,e} \sum_{p\in \P_i: e\in p} x_{i,p}\right)^{\alpha_e} = (1+\frac{\alpha}{\e^\alpha})\opt,\]
which proves the lemma. 
\end{proof}
To make notation simpler, for the analysis we imagine adding dummy resources $E'=\{e' : e\in E\}$ corresponding to the second term in the new objective. We set $\alpha_{e'} := 1$, $c_{e'} := 1$  and $w_{i,e'} := \frac{c_e \alpha_e }{\e^{\alpha}} w_{i,e}^{\alpha_e}$ for all $i\in [N]$ and $e\in E$. Moreover, we extend each reply $p\in \P_i$ so that it contains {\em both} copies $e,e'$ of each resource $e\in p$. The new reply collections are referred to as $\{\P'_i\}_{i=1}^N$.  The dual of the new convex program, denoted $(D')$, is given below.
  \begin{align}
  \max\quad&\sum_{i=1}^N y_i \,-\, \sum_{e\in E\setminus E_1} \frac{ c_e\alpha_e}{\beta_e} \cdot z_e^{\beta_e}\notag\\
  \mbox{s.t.}\quad& \sum_{e\in p}w_{i,e} c_e \alpha_e\cdot z_e \ge y_{i},\qquad \forall p\in \P'_i,\, \forall i\in [N] \label{cons:new-dual}\\
  & z_e \le 1, \qquad\qquad\qquad\qquad\forall e\in E_1\cup E' \label{cons:new-d-UB}\\
  &\mathbf{y}, \mathbf{z} \ge  \mathbf{0}.\notag
\end{align}
Above,  $E_1=\{e\in E: \alpha_e=1\}$. Note that all the dummy resources $E'$ have the exponent $\alpha_e=1$: so they do not appear in the second term of the dual objective. 

Define the dual solution:
  $$z_e := \frac1\rho \cdot A_e^{\alpha_e-1},\qquad \forall e\in E.$$
  $$z_{e'} := 1,\qquad \forall e'\in E'.$$
  $$y_i := \min_{p'\in \P'_i} \sum_{e\in p'} w_{i,e}c_e\alpha_e\cdot z_e= \min_{p\in \P_i} \sum_{e\in p} \left( w_{i,e}c_e\alpha_e\cdot z_e + \frac{c_e \alpha_e }{\e^{\alpha}} w_{i,e}^{\alpha_e}\right),\qquad \forall i\in [N].$$
The second equality above (for $y_i$) follows from the definitions of the new reply-collection $\P_i'$  and weights $w_{i,e'}$, and the setting $z_{e'}=1$ for $e'\in E'$.   Note that dual-constraint~\eqref{cons:new-d-UB} is satisfied as $z_e=\delta\le 1$ for all $e\in E_1$ and $z_{e'}=1$ for all $e'\in E'$. Moreover, \eqref{cons:new-dual} is satisfied by definition of $y$. So $(y,z)$ is a feasible dual solution. For each request $i$, let $q_i\in \P_i$ denote the reply that achieves the minimum cost in the definition of $y_i$ above.
We now relate $A$ with the dual objective $D$.

Consider the algorithm when some request $i$ arrives. Let $\ell_e$ denote the load on each $e\in E$ before request $i$ is assigned. Recall that $p_i\in \P_i$ is the selected reply. Then, the increase in the algorithm's objective, $(\Delta A)_i$ equals:   
\begin{align}
&= \sum_{e\in p_i} c_e \left( (\ell_e+ w_{i,e})^{\alpha_e} - \ell_e^{\alpha_e} \right) \le  \sum_{e\in p_i} c_e   \alpha_e   (\ell_e+ w_{i,e})^{\alpha_e-1} w_{i,e} \label{eq:int-dual-1} \\
&\le\sum_{e\in p_i} c_e   \alpha_e   w_{i,e} \left(\e \cdot \ell_e^{\alpha_e-1} + \alpha_e^{\alpha_e-1} \cdot w_{i,e}^{\alpha_e-1} \right) = \e\cdot \sum_{e\in p_i} \left( c_e   \alpha_e   w_{i,e} \ell_e^{\alpha_e-1} + \frac{1}{\e}c_e\alpha_e^{\alpha_e} w_{i,e}^{\alpha_e} \right) \label{eq:int-dual-2} \\
&\le \e\cdot \sum_{e\in p_i} \left( c_e   \alpha_e   w_{i,e} \ell_e^{\alpha_e-1} + \frac{\rho}{\e^{\alpha}}\cdot c_e \alpha_e w_{i,e}^{\alpha_e} \right) \le \e\cdot \sum_{e\in q_i} \left( c_e   \alpha_e   w_{i,e} \ell_e^{\alpha_e-1} + \frac{\rho}{\e^{\alpha}}\cdot c_e\alpha_e  w_{i,e}^{\alpha_e} \right) \label{eq:int-dual-3}\\
 &=   \e\rho \cdot \sum_{e\in q_i} \left( \frac1\rho \cdot c_e   \alpha_e   w_{i,e} \ell_e^{\alpha_e-1} +  \frac{c_e \alpha_e }{\e^{\alpha}} w_{i,e}^{\alpha_e} \right)\le    \e\rho \cdot \sum_{e\in q_i} \left( \frac1\rho \cdot c_e   \alpha_e   w_{i,e} A_e^{\alpha_e-1} +  \frac{c_e \alpha_e }{\e^{\alpha}} w_{i,e}^{\alpha_e} \right)\label{eq:int-dual-4} \\
&= \e\rho \cdot \sum_{e\in q_i} \left( w_{i,e}c_e\alpha_e\cdot z_e + \frac{c_e \alpha_e }{\e^{\alpha}} w_{i,e}^{\alpha_e}\right)  =  \e\rho \cdot  y_i. \label{eq:int-dual-5}
\end{align} The inequality in \eqref{eq:int-dual-1} uses convexity of the $x^{\alpha_e}$ function. The inequality in \eqref{eq:int-dual-2} uses the inequality $(X+Y)^{\alpha-1} \le \e\cdot X^{\alpha-1} + \alpha^{\alpha-1}\cdot Y^{\alpha-1}$ for $\alpha\ge 1$ and $X,Y\ge 0$, which follows from Lemma~4.1 in \cite{AwerbuchAGKKV95} (by setting $c=\e$). The first inequality in \eqref{eq:int-dual-3} uses $\rho\ge (\e \alpha)^{\alpha-1}$ which we will ensure. The second  inequality in \eqref{eq:int-dual-3} uses the choice of $p_i$ under the costs \eqref{eq:cost-integer-online}. The   inequality in \eqref{eq:int-dual-4} uses the fact that loads are monotonically non-decreasing. The equalities in \eqref{eq:int-dual-5} use the 
definition of reply $q_i$ and choice of dual variables $y_i$ and $z_e$.
 Adding over all $i$,  
$$A \le \e\rho \cdot  \sum_{i=1}^N y_i.$$
Now, consider the contribution of the $z$-variables to the dual objective:
$$\sum_{e\in E\setminus E_1} \frac{c_e\alpha_e}{\beta_e}\cdot z_e^{\beta_e} = \sum_{e\in E\setminus E_1} \rho^{-\beta_e}\frac{c_e\alpha_e}{\beta_e} \left(A_e^{\alpha_e-1}\right)^{\beta_e}= \sum_{e\in E\setminus E_1} \rho^{-\beta_e} c_e(\alpha_e-1)  A_e^{\alpha_e}\le \rho^{-\frac{\alpha}{\alpha-1}}(\alpha-1) \sum_{e\in E\setminus E_1}c_e A_e^{\alpha_e}, $$
which follows the same way as for the fractional online algorithm. 
Therefore, the dual objective is:
$$D = \sum_{i=1}^N y_i - \sum_{e\in E\setminus E_1} \frac{c_e\alpha_e}{\beta_e}\cdot z_e^{\beta_e} \ge \left(\frac{1}{\e \rho} - \rho^{-\frac{\alpha}{\alpha-1}}(\alpha-1) \right)\cdot A.$$
 Finally, choosing $\rho=(\e \alpha)^{\alpha-1}$, we obtain $A\le (\e\alpha)^\alpha\cdot D$. Combined with the observation that  $D\le (1+\alpha\e^{-\alpha})\opt$ (by Lemma~\ref{lem:new-opt}), we obtain Theorem~\ref{thm:online}. 

\subsection{Using Approximate Min-Cost Replies}

\def\tA{\overline{A}}

Here we consider the situation where an exact min-cost reply cannot be computed efficiently. This is indeed the case in some applications. We extend our online algorithm so that it also works with approximately min-cost replies. Moreover, we obtain a stronger guarantee for the linear terms in the objective, which will be used in proving our main result (see \S\ref{sec:overall}). Recall that $E_1\sse E$ denotes the resources with exponent $\alpha_e=1$.
\begin{theorem}\label{thm:apx-online}
Assume that there is a $\tau$-approximation algorithm for the min-cost oracle in \gnd. Then, there is a polynomial time $2(\e \alpha\tau )^\alpha$-competitive online algorithm for \gnd. In fact, if $L$ and $H$ denote the costs incurred by the algorithm on resources in $E_1$ and $E\setminus E_1$ respectively, then $L\le 2\tau\cdot \opt$ and $H\le 2(\e\alpha\tau)^\alpha\cdot \opt$.
\end{theorem}
\begin{proof}
This algorithm is a slight modification of the previous one. Upon arrival of request $i\in [N]$, we select the reply $p_i\in \P_i$ returned by the $\tau$-approximate min-cost   oracle under costs:
\begin{equation}\label{eq:cost-int-LMP}
\bar{\d}_e := \left\{ \begin{array}{ll}
 \alpha_e c_e \cdot \ell_e^{\alpha_e-1}\cdot w_{i,e} \,+\, \frac{\rho}{\e^{\alpha}} \cdot c_e \alpha_e w_{i,e}^{\alpha_e} &\mbox{ if } e\in E\setminus E_1\\
 \rho c_e w_{i,e} & \mbox{ if }e\in E_1
\end{array}\right..
\end{equation}
Note that the only difference from the costs  \eqref{eq:cost-integer-online} in Theorem~\ref{thm:online} is in the cost setting for $E_1$. We will also set the parameter $\rho\ge 1$ differently. We only prove the second statement in the theorem, which clearly implies the first statement.

As before, let $A_e$ denote the algorithm's final load on each resource $e\in E$. Let  $L:=\sum_{e\in E_1} c_e A_e$ and $H:=\sum_{e\in E\setminus E_1} c_e A_e^{\alpha_e}$ denote the costs incurred by the algorithm on resources in $E_1$ and $E\setminus E_1$ respectively. Note that the algorithm's cost $A=L+H$. We will bound the modified objective $\tA:=\e \rho\cdot L +H$. In the analysis, we will use a slightly different relaxation for \gnd.  The new relaxation has the same constraints in $(P)$ with objective:
\begin{equation}
    \label{eq:improved-obj}
\sum_{e\in E} c_e\cdot \left( \sum_{i=1}^N w_{i,e} \sum_{p\in \P_i: e\in p} x_{i,p}\right)^{\alpha_e}\,\,+\,\,\sum_{e\in E\setminus E_1} \frac{c_e \alpha_e }{\e^{\alpha}} \sum_{i=1}^N w_{i,e}^{\alpha_e}\sum_{p\in \P_i: e\in p} x_{i,p} \end{equation}
As in Lemma~\ref{lem:new-opt}, the optimal value of this new  convex program is at most $(1+\alpha\e^{-\alpha})\cdot \opt$, where \opt is the optimal value of the (integral) \gnd instance. 
We imagine adding dummy resources $E''=\{e' : e\in E\setminus E_1\}$ corresponding to the second term in the new objective \eqref{eq:improved-obj}. We set $\alpha_{e'} := 1$, $c_{e'} := 1$  and $w_{i,e'} := \frac{c_e \alpha_e }{\e^{\alpha}} w_{i,e}^{\alpha_e}$ for all $i\in [N]$ and $e\in E\setminus E_1$. Moreover, we extend each reply $p\in \P_i$ so that it contains {\em both} copies $e,e'$ of each resource $e\in p\setminus E_1$. The new reply collections are referred to as $\{\P''_i\}_{i=1}^N$.  The dual of the new convex program, denoted $(D'')$, is given below.
  \begin{align*}
  \max\quad&\sum_{i=1}^N y_i \,-\, \sum_{e\in E\setminus E_1} \frac{ c_e\alpha_e}{\beta_e} \cdot z_e^{\beta_e} \\
  \mbox{s.t.}\quad& \sum_{e\in p''}w_{i,e} c_e \alpha_e\cdot z_e \ge y_{i},\qquad \forall p''\in \P''_i,\, \forall i\in [N]  \\
  & z_e \le 1, \qquad\qquad\qquad\qquad\forall e\in E_1\cup E''  \\
  &\mathbf{y}, \mathbf{z} \ge  \mathbf{0}. 
\end{align*}
 Note that $\alpha_{e'}=1$ for all dummy resources $e'\in E''$: so they do not appear in the second term of the dual objective. 
We define the following dual solution:
  $$z_e := \frac1\rho \cdot A_e^{\alpha_e-1},\qquad \forall e\in E\setminus E_1.$$
  $$z_{e'} := 1,\qquad \forall e'\in E_1\cup E''.$$
  $$y_i := \min_{p''\in \P''_i} \sum_{e\in p''} w_{i,e}c_e\alpha_e\cdot z_e= \min_{p\in \P_i} \left( \sum_{e\in p \cap E_1}c_e w_{i,e} + \sum_{e\in p\setminus E_1} \left( c_e\alpha_e w_{i,e}\cdot z_e + \frac{c_e \alpha_e }{\e^{\alpha}} w_{i,e}^{\alpha_e}\right)\right),\quad \forall i\in [N].$$
The only difference from the choice in Theorem~\ref{thm:online} is that now $z_e=1$ for all $e\in E_1\cup E''$. It is easy to see that this is feasible for the dual program $(D'')$.  Let $q_i\in \P_i$ denote the reply that achieves the min-cost in the definition of $y_i$.

Now consider the increase in $\tA$ when request $i$ arrives. Let $\ell_e$ denote the load on $e\in E$ before request $i$ is assigned. 
 \begin{align}
(\Delta \tA)_i &= \e \rho \sum_{e\in p_i\cap E_1} c_e w_{i,e} \,+\, \sum_{e\in p_i\setminus E_1} c_e\left( (\ell_e+ w_{i,e})^{\alpha_e} - \ell_e^{\alpha_e} \right) \notag\\
&\le \e \rho \sum_{e\in p_i\cap E_1} c_e w_{i,e} \,+\,  \e\cdot \sum_{e\in p_i\setminus E_1} \left( c_e   \alpha_e   w_{i,e} \ell_e^{\alpha_e-1} + \frac{\rho}{\e^{\alpha}}\cdot c_e \alpha_e w_{i,e}^{\alpha_e} \right) \label{eq:improved-2}\\
&= \e \sum_{e\in p_i} \bar{\d}_e \le \e \tau \sum_{e\in q_i} \bar{\d}_e = \e \tau \rho \left(\sum_{e\in q_i\cap E_1} c_e w_{i,e} +   \sum_{e\in q_i\setminus E_1} \left( \frac{c_e   \alpha_e   }\rho w_{i,e} \ell_e^{\alpha_e-1} + \frac{c_e \alpha_e }{\e^{\alpha}} w_{i,e}^{\alpha_e} \right)\right) \label{eq:improved-3}\\
&\le \e \tau \rho \left(\sum_{e\in q_i\cap E_1} c_e w_{i,e} +   \sum_{e\in q_i\setminus E_1} \left( \frac1\rho c_e   \alpha_e   w_{i,e} A_e^{\alpha_e-1} + \frac{c_e \alpha_e }{\e^{\alpha}} w_{i,e}^{\alpha_e} \right)\right)  =  \e\tau \rho \cdot  y_i. \label{eq:improved-4}
\end{align} 
Inequality \eqref{eq:improved-2} follows by the same calculations as in \eqref{eq:int-dual-1}-\eqref{eq:int-dual-3}. The equalities in \eqref{eq:improved-3} is by definition of the costs~\eqref{eq:cost-int-LMP}, and the inequality in \eqref{eq:improved-3} is by choice of $p_i$. The inequality in  \eqref{eq:improved-4} is by the monotonicity of loads. 
So we obtain $\tA \le \e\tau\rho \sum_{i=1}^N y_i$. 
Moreover, we can bound the contribution of the $z$-variables in the exact same way as before, to obtain: 
$$\sum_{e\in E\setminus E_1} \frac{c_e\alpha_e}{\beta_e}\cdot z_e^{\beta_e}  \le \rho^{-\frac{\alpha}{\alpha-1}}(\alpha-1) \sum_{e\in E\setminus E_1}c_e A_e^{\alpha_e} = \rho^{-\frac{\alpha}{\alpha-1}}(\alpha-1) \cdot H.$$ 
Therefore, the dual objective is:
$$D = \sum_{i=1}^N y_i - \sum_{e\in E\setminus E_1} \frac{c_e\alpha_e}{\beta_e}\cdot z_e^{\beta_e} \ge \frac{\tA}{\e\tau \rho} - \rho^{-\frac{\alpha}{\alpha-1}}(\alpha-1) \cdot H = \frac{L}{\tau} + \left(\frac{1}{\e \tau\rho} - \rho^{-\frac{\alpha}{\alpha-1}}(\alpha-1)  \right)H. $$
 Finally, choosing $\rho=(\e\tau \alpha)^{\alpha-1}$, we obtain: 
 $$\frac{L}{\tau} \,+\, \frac{H}{(\e\tau\alpha)^\alpha} \le D\le (1+\alpha\e^{-\alpha})\opt,$$ where we used that the optimal value of $(D'')$ is at most $(1+\alpha\e^{-\alpha})\opt\le 2\opt$. Hence, $L\le 2\tau\cdot \opt$ and $H\le 2(\e\alpha\tau)^\alpha\cdot \opt$, which   completes the proof. 
\end{proof}

\section{Application to \gnd with (D)oS Costs}\label{sec:overall}
\def\I{\ensuremath{{\cal I}}\xspace}
\def\J{\ensuremath{{\cal J}}\xspace}

We now complete the proof of our main result (Theorem~\ref{thm:main-1}). Given an instance \I  of \gnd with (D)oS costs as in \eqref{eq:energy}, we use Lemma~\ref{lem:fn-apx} to define a new instance \J of \gnd with power cost functions, as follows. For each original resource $e\in E$, we have two copies $e_1$ and $e_a$. Let $E_1:=\{e_1:e\in E\}$ and $E_a:=\{e_a:e\in E\}$: so the resources in \J are $E'=E_1\cup E_a$. Define scalars $c_{e_1}:=\xi_e q_e^{\alpha_e-1}$ and $c_{e_a}:=\xi_e$ for all $e\in E$. Also, define exponents $\alpha_{e_1}:=1$ and $\alpha_{e_a}:=\alpha_e$ for all $e\in E$. The weighted power functions in instance \J are $g_{r}(x):= c_{r}\cdot  x^{\alpha_{r}}$ for all resources $r\in E'$.  The reply-collections are extended so that for each reply $p\in \P_i$ in \I, there is a corresponding reply in \J that contains {\em both} copies of resources $e\in p$. For each $e\in E$, note that function $h_e(x)$ used in Lemma~\ref{lem:fn-apx} is $h_e(x) = g_{e_1}(x)+g_{e_a}(x)$. Using the first inequality in Lemma~\ref{lem:fn-apx}, the optimal value of instance \J is $\opt_\J \le 2\cdot \opt_\I$. Now, using Theorem~\ref{thm:apx-online} on instance \J, we obtain:
$$L=\sum_{r\in E_1} c_r A_r \le 2\tau\cdot \opt_\J \mbox{ and } H=\sum_{r\in E_a} c_r A_r^{\alpha_r} \le 2(\e\tau\alpha)^\alpha\cdot \opt_\J,$$
where $\{A_r\}_{r\in E'}$ denote  the final loads in the algorithm. For each $e\in E$, note that $A_{e_1}=A_{e_a}$; we use $A_e$ to denote this load. As every weight is at least one, we have each $A_e\in \{0\}\cup \R_{\ge 1}$. The objective value in the original instance \I is
\begin{align}
\sum_{e\in E} f_e(A_e) &\le \sum_{e\in E} \left( \max\{q_e,1\}\cdot \xi_e q_e^{\alpha_e-1} \cdot A_e + \xi_e \cdot A_e^{\alpha_e}\right)\label{eq:combine-0} \\
&=\sum_{e\in E} \left( \max\{q_e,1\}\cdot c_{e_1}\cdot A_e + c_{e_a}\cdot A_e^{\alpha_e}\right)\le \max\{q,1\}\cdot L + H \label{eq:combine-1} \\
&\le 2\left(\max\{q,1\}\tau+ (\e\tau\alpha)^\alpha\right)\cdot \opt_\J\le 4\left(\max\{q,1\}\tau+ (\e\tau\alpha)^\alpha\right)\cdot \opt_\I.\label{eq:combine-2}
\end{align}
Inequality \eqref{eq:combine-0} is by Lemma~\ref{lem:fn-apx} (2nd inequality) and $A_e\in \{0\}\cup \R_{\ge 1}$.
In \eqref{eq:combine-1}, the equality is by definition of the scalars $c_r$ and the inequality is by definition of $L$ and $H$. In \eqref{eq:combine-2}, the first  inequality is by the above bounds on $L$ and $H$, and the last inequality  uses $\opt_\J \le 2\cdot \opt_\I$. This completes the proof of Theorem~\ref{thm:main-1}.

\paragraph{Remark:} The requirement that every weight is at least one is crucial in obtaining our result. As noted earlier, this requirement also appears in all prior work, e.g. \cite{AndrewsAZZ12,AndrewsAZ16,AntoniadisIKMNP20,MakarychevS18,EmekKLS20}. In fact,  any $r(q)$ competitive ratio for \gnd under arbitrary weights (possibly less than one) leads to an $O(1)$-competitive online algorithm, which is not possible even for the simplest setting of single-commodity flow in edge-weighted undirected graphs.
 To see this, consider a new instance of \gnd with weights $w'_{i,e}=w_{i,e}/q$ and parameters $\sigma'_e=\sigma_e/q^\alpha$ and $\xi'_e=\xi_e$. Note that the new \gnd instance is equivalent to the old one (the objective value of each solution is scaled down by $q^\alpha$). Moreover, the new value $q'=1$, which means that we have an $r(q')=O(1)$ competitive algorithm. 
\paragraph{REP cost functions} We now consider the \gnd problem under more general costs of the form~\eqref{eq:REP} and prove Theorem~\ref{thm:main-3}. The main idea is to replace each resource $e\in E$ with $q$ {\em copies} $e_1,\cdots e_q$ each with (D)oS cost function of the usual form~\eqref{eq:energy}. Then, we will directly apply Theorem~\ref{thm:main-1}. 

For each $e\in E$, (by renumbering if needed) let $$\left(\frac{\sigma_e}{\xi_{e,1}}\right)^{1/\alpha_{e,1}}\,\, =\,\, \min_{j=1}^q  \left(\frac{\sigma_e}{\xi_{e,j}}\right)^{1/\alpha_{e,j}}.$$   
The new \gnd instance has resources $\bar{E}:=\{e_j : j\in [q], e\in E\}$. For each $e\in E$, set 
$$\sigma_{e_j}:= \left\{ \begin{array}{ll}
\sigma_e     &  \mbox{ if }j=1\\
     0 & \mbox{ if } j=2,\cdots q
\end{array}\right., \mbox{ and } \alpha_{e_j}:=\alpha_{e,j}, \xi_{e_j}:=\xi_{e,j} \mbox{ for all } j\in [q].$$ Let $f_{e_j}(x)$ denote the (D)oS cost function for each $e_j\in \bar{E}$. Clearly, $\bar{f}_e(x) = \sum_{j=1}^q f_{e_j}(x)$ for all $x\ge 0$ and $e\in E$. Moreover,$$q:=\max_{f\in \bar{E}} \left(\frac{\sigma_f}{\xi_f}\right)^{1/\alpha_f} = \max_{e\in E} \min_{j\in [q]} \left(\frac{\sigma_e}{\xi_{e,j}}\right)^{1/\alpha_{e,j}} = Q.$$
Recall the definition of $Q$ in Theorem~\ref{thm:main-3}. For each request $i$, the new reply-collection  is $$\bar{\P}_i:=\{\cup_{e\in p} \{e_1,\cdots e_q\} \,:\, p\in \P_i\},$$
i.e., each new reply corresponds to selecting all copies of the resources in some original reply $p$. Assuming a $\tau$-approximation algorithm for the min-cost oracle under $\P_i$, it is easy to obtain a $\tau$-approximation algorithm for the new min-cost oracle under $\bar{\P}_i$. Indeed, given costs $d:\bar{E}\rightarrow \R_+$ we define costs $d' : E\rightarrow \R_+$ as $d'_e:=\sum_{j=1}^q d_{e_j}$ for each $e\in E$ and apply the oracle for $\P_i$. 

Using Theorem~\ref{thm:main-1} on the new \gnd instance, we obtain an $O(q\tau+(\e \alpha\tau)^\alpha) = O(Q\tau+(\e \alpha\tau)^\alpha)$ competitive online algorithm under REP cost functions. This proves Theorem~\ref{thm:main-3}. The runtime of this algorithm is $O(Nmq + N\Phi(mq))$ where $\Phi(\cdot)$ denotes the time taken by the min-cost oracle.

\section{Lower Bounds}\label{sec:LB}
\def\ssr{\ensuremath{\mathsf{SSR}}\xspace}

We now show that our competitive ratio is tight up to a constant factor and prove Theorem~\ref{thm:main-2}.

We consider the single commodity routing problem (\ssr) in directed graphs, which is a special case of \gnd. We are given a directed graph $(V,E)$ with weight  $c_e\ge 0$ associated with each edge $e\in E$. There is a common source $s\in V$ and each online request $i$ corresponds to routing unit flow from $s$ to a sink node $t_i\in V$. The edge cost function of each edge is $f_e(x) = c_e\cdot f(x)$ where
$$f(x) = \left\{ \begin{array}{ll}
0 &\mbox{ if }x=0\\
\sigma+x^\alpha&\mbox{ if }x>0
\end{array} 
\right..
$$
Note that $q=\sigma^{1/\alpha}$. The min-cost reply oracle corresponds to shortest path: so we also have a polynomial time exact oracle in this case. We provide two different instances of \ssr that show lower bounds of (i) $\Omega(q)$ for every choice of $\alpha\ge 1$  and $\sigma\ge 0$, and (ii) $\Omega((1.44\alpha)^\alpha)$ even when $q=0$.

The $\Omega((1.44\alpha)^\alpha)$ lower bound follows from the restricted-assignment scheduling problem with $\ell_p$-norm of loads~\cite{Caragiannis08}. Recall, in that problem there are $m$ machines and $N$ jobs arrive over time. Each job $i$ specifies a subset $M_i$ of machines and needs to be assigned to one of them. The objective is to minimize the sum of $p^{th}$ powers of the machine loads. This corresponds to the directed graph on nodes $\{s\}\cup\{u_e\}_{e=1}^m\cup \{v_i\}_{i=1}^N$ where $s$ is the source, $u$-nodes correspond to machines and $v$-nodes correspond to jobs. There is an edge from $s$ to each $u_e$ with weight $1$. For each job $i\in [N]$ and machine $e\in M_i$, there is an edge from $u_e$ to $v_i$ of weight zero. We also set $\alpha=p$ and $\sigma=0$. The resulting \ssr instance is clearly equivalent to the scheduling problem. 

The $\Omega(q)$ lower bound uses a construction similar to the lower bound for online directed Steiner tree~\cite{FaloutsosPS02}. Fix any value of $\sigma>0$ and $\alpha\ge 1$ (which also fixes $q$). We will show an $\Omega(q)$ lower bound for \ssr instances with this value of $\sigma$ and $\alpha$. The graph $G$ consists of a complete binary tree $B$ of depth $q$ rooted at node $t$ with all edges directed towards  $t$, and source $s$ with edges to all nodes of the tree. Let $S$ denote all the edges out of $s$. All edges of the binary tree have weight zero and all edges in $S$ have weight one.  The input sequence consists of $q$ requests as follows. At any point in the algorithm, let $A$ denote all edges that carry flow at least one: so the current cost is at least $|A\cap S|\cdot \sigma$. The first sink $t_1=t$. For $i=2,\cdots q$, sink $t_i$ is chosen to be the child of $t_{i-1}$ in $B$ such that $A$ does not contain an $s-t_i$ path. It is clear that $|A\cap S|\ge q$ at the end of this request sequence. So the online cost is at least $q\sigma$. Note that the sinks $t_1, \cdots t_q$ lie on a single directed path in the tree $B$: so an offline  solution can just select the edges $(s,t_q)$ followed by $(t_i, t_{i-1})$ for $i=q,\cdots 2$. The cost of this solution is at most $\sigma+q^\alpha=2\sigma$ as it uses only one edge in $S$ (which carries flow of $q$). Thus, the competitive ratio is at least $q/2$. 

{\small 
\bibliographystyle{alpha}
\bibliography{online-GND}
}

\appendix

\section{Deriving the Dual Program}\label{app:derive-dual}
\def\we{\ensuremath{\mathsf{w}_e}\xspace}
The primal program $(P)$ is:
\begin{align*}
  \min\quad&\sum_{e\in E} c_e\cdot \left( \sum_{i=1}^N w_{i,e} \sum_{p\in \P_i: e\in p} x_{i,p}\right)^{\alpha_e}\\
  \mbox{s.t.}\quad& \sum_{p\in \P_i}x_{i,p} \ge 1,\qquad \forall i\in [N] \\
  &\mathbf{x} \ge  \mathbf{0}.
\end{align*}
Let $k=\sum_{i=1}^N|\P_i|$ denote the number of variables in this convex program. Let $A \in \R_{N\times k}$ denote the constraint matrix for the covering constraints; note that $A$ is block-diagonal with
$$a_{j, (i,p)}=\left\{\begin{array}{ll}
1 &\mbox{ if }i=j\\
0 & \mbox{ if } i\ne j
\end{array}\right., \qquad \forall j\in [N],\, i\in [N],\, p\in \P_i.$$
For simpler notation, let $\we\in \R^k$ denote the vector with $\we(i,p)=w_{i,e}$ if $p\ni e$ (and $0$ otherwise) for all  
$p\in \P_i$ and $i\in [N]$. So $\we^T x = \sum_{i=1}^N w_{i,e} \sum_{p\in \P_i: e\in p} x_{i,p}$. Define functions
$$g_e(x) := c_e\cdot \left(\we^T x\right)^{\alpha_e},\,\forall e\in E, \quad \mbox{ and } \quad g(x) := \sum_{e\in E} g_e(x).$$
Note that $g(x)$ is the objective in our convex program.  Letting $\{y_i\}_{i=1}^N$ denote the Lagrange multipliers for the covering constraints, we obtain the Lagrangian function:
\[L(y)=L(y_1, ...y_N) = \inf_{x \geq 0} \left( g(x) + y^T (\mathbf{1} - Ax)\right) = y^T \mathbf{1} - g^*(A^Ty),\]
where $g^*(\mu) := \sup_{x\ge 0} (\mu^T x - g(x))$ is the Fenchel conjugate of $g(x)$. The dual program is then:
\begin{equation}
\label{eq:app-dual}
\sup_{y\ge 0} L(y) = \sup_{y\ge 0} \left(y^T \mathbf{1} - g^*(A^Ty)\right).\end{equation}
Note that we always have $\mu:=A^Ty\ge 0$ as $y\ge 0$ and $A$ is non-negative.

As $g(x) = \sum_{e \in E} g_e(x)$, we can use the Moreau-Rockafeller formula to compute \[g^*(\mu) = \bigoplus_{e\in E} g_e^*(\mu) = \inf_{\sum_{e } \mu_e = \mu} \sum_{e\in E}g_e^*(\mu_e),\]
where $\oplus$ is the infimal convolution. 
We now compute the conjugate functions for $g_e$s depending on whether $\alpha_e>1$ or $\alpha
_e=1$. Recall that $E_1=\{e\in E:\alpha_e=1\}$.

\begin{cl}\label{cl:conj>1} For any $e\in E$ with $\alpha_e>1$ and $\mu\ge 0$,
$$g_e^*(\mu) =  
c_e(\alpha_e-1)\left(\frac{\lambda}{c_e\alpha_e}\right)^{\beta_e}, \quad \mbox{ where }\lambda = \max_{i,p} \frac{\mu(i,p)}{\we(i,p)} \mbox{ and }\beta_e=\frac{\alpha_e}{\alpha_e-1}.$$
Above, we treat $0/0=0$. Also, $g^*_e(\mu)=\infty$ if $\mu(i,p)>0$ for any $(i,p)$ with $\we(i,p)=0$. \end{cl}
\begin{proof}
Note that $g_e^*(\mu) = \sup_{x\ge 0} \left(\mu^T x - c_e(\we^Tx)^\alpha_e\right)$. For simpler notation, we use $h$ to index the coordinates $(i,p)$ in the vectors $x$ and $\we$.  

We first show that if the supremum in $g_e^*(\mu)$ is achieved, there is an optimal $x$ with at most one positive coordinate. Suppose, that $x$ achieves the supremum in $g_e^*(\mu)$ and has $x_1,x_2>0$. Without loss of generality, say $\frac{\mu(1)}{\we(1)} \ge  \frac{\mu(2)}{\we(2)}$. Then, for some $\epsilon>0$, consider the new solution $x'$ with $x'_1=x_1+\epsilon \we(2)$, $x'_2=x_2-\epsilon \we(1)$  and  $x'_h=x_h$ for other coordinates $h$. The increase in the objective is $\epsilon (\mu(1)\we(2) - \mu(2)\we(1)) \ge 0$. So we can choose an  $\epsilon>0$ such that $x'_2=0$ and the objective at $x'$ is at least as much as at $x$. Repeating this process, we will end up with a solution with at most one positive coordinate, as claimed. 

By the above argument, we can assume that an optimal $x$ (if any) will only have positive value on the coordinate $h$ that maximizes $\mu(h)/\we(h)$. By renumbering let $h=1$ denote this coordinate. Then $g_e^*(\mu) = \sup_{x_1\ge 0} \left(\mu(1)\cdot x_1 - c_e (\we(1)\cdot x_1)^{\alpha_e}\right)$. By simple calculus, if $\we(1)=0$ and $\mu(1)>0$ then $g_e^*(\mu)=\infty$; otherwise 
the supremum is achieved and $g_e^*(\mu)=c_e(\alpha_e-1)\left(\frac{\mu(1)}{\we(1) c_e\alpha_e}\right)^{\frac{\alpha_e}{\alpha_e-1}}$. The claim follows by definition of $\lambda$ and $\beta_e$.
\end{proof}

Using Claim~\ref{cl:conj>1}, for any $e\in E\setminus E_1$ we can re-write:
\begin{equation}\label{eq:app-dual-1}
g^*_e(\mu_e) = \min \left\{ \frac{c_e\alpha_e}{\beta_e} \left( \frac{\lambda_e}{c_e\alpha_e}\right)^{\beta_e} \,:\, \we(i,p)\cdot \lambda_e \ge \mu_e(i,p) \mbox{ for all }(i,p)\right\},
\end{equation}
where an infeasible problem has value $\infty$. 

\begin{cl}\label{cl:conj=1} For any $e\in E$ with $\alpha_e=1$ and $\mu\ge 0$,
$$g_e^*(\mu) = \left\{ \begin{array}{ll}
0& \mbox{ if }\mu(i,p) \le c_e\cdot \we(i,p) \mbox{ for all } i\in [N] \mbox{ and } p\in \P_i\\
\infty & \mbox{ otherwise}.
\end{array}\right..$$
\end{cl}
\begin{proof}
Note that $g_e^*(\mu) = \sup_{x\ge 0} \left(\mu^T x - c_e(\we^T x)\right)$. Again, we use $h$ to index the coordinates $(i,p)$ in the vectors $x$ and $\we$.  
First, suppose $\mu(h)>c_e\cdot \we(h)$ for some $h$ (say $h=1$). Then, setting $x_h=0$ for $h\ne 1$ and $x_1\rightarrow \infty$ we obtain $g_e^*(\mu)=\infty$. 

Now, we assume $\mu \le c_e\cdot \we$. Then, it is clear that $g_e^*(\mu)=0$.
\end{proof}
 Using Claim~\ref{cl:conj=1}, for $e\in E_1$, we have: 
\begin{equation}\label{eq:app-dual-2}
g^*_e(\mu_e) = \min \left\{ 0 \,:\, \we(i,p)\cdot \lambda_e \ge \mu_e(i,p) \mbox{ for all }(i,p),\, \mbox{ and }\lambda_e\le c_e\right\},
\end{equation}
where an infeasible problem has value $\infty$.

By Claims~\ref{cl:conj>1} and \ref{cl:conj=1}, we know that $g^*_e(\mu)$ is non-decreasing in $\mu$. So we can write:
\[g^*(\mu) =    \inf_{\sum_{e } \mu_e \ge \mu} \sum_{e\in E}g_e^*(\mu_e).\]
Combined with \eqref{eq:app-dual-1} and \eqref{eq:app-dual-2}, we have:
 \begin{align*}
  g^*(\mu) = \quad &\min\quad \sum_{e\in E\setminus E_1}  \frac{c_e\alpha_e}{\beta_e} \left( \frac{\lambda_e}{c_e\alpha_e}\right)^{\beta_e} \\
  &\mbox{s.t.}\quad  \sum_{e\in E} \lambda_e \cdot \we(i,p) \ge \mu(i,p),\qquad \forall i\in [N],\, p\in \P_i, \\
   &\qquad\lambda_e\le c_e,\qquad \forall e\in E_1,\\
   &\qquad\mathbf{\lambda} \ge  \mathbf{0}.
\end{align*}
Combining this with \eqref{eq:app-dual}, using $\mu=A^T y$ and the definition of matrix $A$, the dual program is:
 \begin{align*}
 \max&\quad \sum_{i=1}^N y_i \, -\, \sum_{e\in E\setminus E_1}  \frac{c_e\alpha_e}{\beta_e} \left( \frac{\lambda_e}{c_e\alpha_e}\right)^{\beta_e} \\
   \mbox{s.t.}&\quad  \sum_{e\in E} \lambda_e \cdot \we(i,p) \ge y_i,\qquad \forall i\in [N],\, p\in \P_i, \\
   &\qquad\lambda_e\le c_e,\qquad \forall e\in E_1,\\
   &\qquad\mathbf{\lambda}, \mathbf{y}  \ge  \mathbf{0}.
\end{align*}
Using  new variables $z_e=\lambda_e/(c_e\alpha_e)$ we obtain the dual program $(D)$ as described in \S\ref{sec:frac}.

\end{document}